\documentclass[11pt,times]{article}
\usepackage{amsfonts}
\usepackage{amssymb,amsmath,amsthm}
\usepackage{epsfig}
\usepackage{color}
\usepackage{latexsym}
\usepackage{enumerate}
\usepackage{fullpage,setspace}
\usepackage[dvips, paper=letterpaper, top=0.95in, bottom=.75in, left=0.9in, right=0.95in, nohead, includefoot, footskip=.25in]{geometry}
\usepackage{bbm}

\usepackage{sgame}

\usepackage{verbatim}

\usepackage{pstricks}
\usepackage{egameps}

\newtheorem{lem}{Lemma}
\newtheorem{claim}{Claim}
\newtheorem*{lem*}{Lemma}
\newtheorem{thm}{Theorem}
\newtheorem*{thm*}{Theorem}
\newtheorem{cor}{Corollary}
\newtheorem{mydef}{Definition}
\newtheorem*{mydef*}{Definition}

\newcommand{\ket}[1]{\left| #1 \right>} 
\newcommand{\bra}[1]{\left< #1 \right|} 

\begin{document}
\title{Quantum Correlated Equilibria in Classical Complete Information Games}
\author{Alan Deckelbaum\thanks{Department of Mathematics, MIT. Email: deckel@mit.edu. Supported by Fannie and John Hertz Foundation, Professor Daniel Stroock Fellowship.}}

\maketitle
\thispagestyle{empty}



\section{Introduction}

Since the work of Aumann \cite{aumann}, the concept of correlated equilibrium (CE) has played an important role in the study of games. Correlated equilibria always exist, and unlike Nash equilibria, which are believed to be computational intractable (see \cite{daskalakis}), a correlated equilibrium can be computed efficiently in a broad class of succinctly-representable games \cite{papadimitriou}. In a correlated equilibrium, a trusted correlating device selects strategies  from a joint probability distribution and privately sends a recommended move to each player. Each player maximizes his expected utility by following his recommendation. The question of how to implement a correlated equilibrium without a trusted third party has recently attracted the attention of the cryptographic community. For example, it has been studied (see \cite{dodis}, \cite{ilm1}, \cite{ilm2}, \cite{lepinski}, \cite{urbano}) how to use cryptographic protocols to replace the trusted mediator with multiple rounds of interaction between players.

In this paper, we study the scenario where, instead of having access to a mediator or the ability to perform cryptographic protocols via rounds of communication, the players of a classical complete information game initially share an entangled \textit{pure} quantum state. Each player may perform arbitrary local operations on his own qubits (by using the state as input to an arbitrary quantum circuit) in order to determine which move to play, but no direct communication between players is allowed. An appealing point of our model is that it has a simple theoretical implementation: assuming that players have access to disjoint qubits from an appropriate quantum state, we do not need any interaction between the players (in the form of a cryptographic protocol) or any communication with a trusted mediator. 

In our framework, we define the concept of \textit{quantum correlated equilibrium} (QCE) for both normal and extensive form games of complete information. We show that in a normal form game, any outcome distribution implementable by a QCE can also be implemented by a classical CE. We prove that the converse is surprisingly false: We give an example of an outcome distribution of a normal form game which is implementable by a CE, yet we prove that in any attempted quantum protocol achieving this distribution, at least one of the players will have incentive to deviate.

We extend our analysis to extensive form games, and find that the relation between classical and quantum correlated equilibria becomes less clear. We compare the outcome distributions implementable in our quantum model to those implementable by a classical \textit{extensive form correlated equilibrium} (EFCE) (see von Stengel and Forges \cite{vonstengel}).\footnote{Since we are only concerned with games of complete information, we can avoid many of the technicalities from \cite{forges} and \cite{vonstengel}.} For example, we show that there exists an extensive form complete information game and a distribution of outcomes which can be implemented by a QCE but not by any EFCE, in contrast to the result for normal form games. We also consider the concept of an \textit{immediate-revelation extensive form correlated equilibrium} (IR-EFCE) (motivated by discussion in Forges \cite{forges2}) and compare the power of IR-EFCE to EFCE and to QCE.

\subsection{Related Work}

While work by Clauser et al \cite{chsh}, Cleve et al \cite{cleve} and La Mura \cite{lamura} have studied how quantum entanglement can aid in games of incomplete information (such as Bayesian games), we restrict our attention to games of complete information, and find that even in this framework the questions are nontrivial. Quantum solutions of classical coordination games have been studied previously, such as in Cleve et al \cite{cleve} and Huberman et al \cite{huberman}. In this paper, we look at games which have both cooperative and competitive components. Instead of analyzing the ``quantization'' of games (see Meyer \cite{meyer}), our underlying games remain purely classical. Entanglement is used only as a device to aid in a player's decision of which strategy to play in the classical game. By keeping the underlying game classical, our model generalizes naturally from normal form to extensive form games.

Since our goal is to study a mediator-free setting, it is necessary to restrict our model so that the initial shared state be pure (See Appendix~\ref{purestates}). This restriction is very significant and differs from work such as Zhang's \cite{shengyu} which, while studying both pure and mixed initial states, limited its mention of pure states to those with a certain restricted form.\footnote{Roughly speaking, the main difference is that we allow for pure states with many ancillary qubits, and these ancillary qubits can indeed affect the players' ability to gain utility by deviating.} Furthermore, unlike La Mura's model \cite{lamura}, in our definition of equilibrium we do not restrict the local operations that a player might potentially perform to his own qubits. 


\section{Classical Correlated Equilibria}

We briefly discuss the concept of correlated equilibrium in classical complete information games, and elaborate on this concept in Appendix~\ref{normalform}. For a more thorough discussion, see \cite{aumann}, \cite{forges}, or \cite{vonstengel}. 

\subsection{Normal Form Games}

Correlated equilibrium (CE) in normal form games was first introduced by Aumann \cite{aumann}. In a correlated equilibrium of a normal form game, a trusted \textit{correlating device} selects an outcome of the game according to some known probability distribution, and privately suggests to each player the appropriate action to achieve this outcome. The resulting play is a CE if each player can maximize his expected utility by always following his recommendation, given that all other players follow their recommendations. See Appendix~\ref{normalform} for further discussion and an example, or see \cite{aumann} for formal definitions.

\subsection{Extensive Form Games}

We informally present the concept of classical correlated equilibrium in extensive form games of complete (but imperfect) information, following \cite{vonstengel}. A more thorough discussion can be found in \cite{forges} and \cite{vonstengel}. Note that there are several different ways of defining correlated equilibria in extensive form games, and in this section we present two such versions.

An extensive form game $G$ has a finite set of players, $n$. The game is represented as a rooted directed tree, where the non-terminal nodes are partitioned into information sets. Each information set belongs to a single player.\footnote{Throughout this paper we will assume that the game has the perfect recall property.} A \textit{pure strategy} for player $i$ selects a single outgoing edge from every information set belonging to $i$. Denote the set of pure strategies available to player $i$ by $\Sigma_i$.

A \textit{correlating device} $\mu$ is a distribution over $\prod_{i \in n} \Sigma_i$. Consider the following procedure:
\begin{itemize}
	\item A trusted mediator draws a strategy profile $\pi = (\pi_1, \pi_2, \ldots, \pi_n)$ according to the correlating device $\mu$.
	\item The players begin playing the game $G$. As the gameplay enters each information set, the mediator tells the set's owner $i$ the recommended move according to $\pi_i$. 
\end{itemize}

Following \cite{vonstengel}, we say that $C = (G, \mu)$ is an \textit{extensive form correlated equilibrium} (EFCE) if, for every player $i$, given that all other players follow their recommended move, player $i$'s expected utility is maximized by always following his recommendation.\footnote{We do not impose any requirement of subgame perfection in our equilibrium definition.}

In the above protocol, the strategy profile $\pi$ defines a suggested move at every information set. This recommendation is revealed only to the set's owner, and is only revealed when he reaches the set. In addition to the definition from \cite{vonstengel}, we give an alternate definition of extensive form correlated equilibrium (briefly mentioned in \cite{forges2}), which we will call \textit{immediate-revelation extensive form correlated equilibrium} (IR-EFCE) defined analogously to that above except where player $i$ learns his entire strategy recommendation $\pi_i$ before gameplay begins. We compare the various classical correlated equilibrium concepts in Appendix~\ref{comparison}. 

\section{Quantum Correlated Equilibria in Normal Form Games}\label{QCE}

In this section we discuss the concept of a quantum correlated equilibrium (QCE) in normal form games.

\begin{mydef}
Let $G$ be a normal form game with $n$ players. For each player $i$, let $A_i$ be the set of actions available to player $i$ in $G$. Consider a 3-tuple $(\ket{\psi}, \Gamma, Q)$ where
\begin{itemize}
	\item $\ket{\psi}$ is pure quantum state.
	\item $\Gamma$ is a partition of the qubits of $\ket{\psi}$ into $n$ disjoint sets $q_1, q_2, \ldots, q_n$.
	\item $Q = (Q_1, \ldots, Q_n)$ is a collection of $n$ quantum circuits, where circuit $Q_i$ takes as input the qubits $q_i$ (as well as auxiliary $\ket{0}$ qubits) and outputs an action $a_i \in A_i$. 
\end{itemize}
Given such a 3-tuple, we denote $D(\ket{\psi}, \Gamma, Q)$ as the distribution resulting over outcomes of $G$ when each player $i$ applies $Q_i$ to his qubits of $\ket{\psi}$ and plays the result, and let $u_i(D)$ be the expected utility for player $i$ in the outcome distribution $D$. 

We say that $(\ket{\psi}, \Gamma, Q)$ is a \textbf{quantum correlated equilibrium} (QCE) if, for all players $i$ and for all quantum circuits $Q_i'$
$$u_i(D(\ket{\psi}, \Gamma, Q)) \geq u_i\left(D\left(\ket{\psi}, \Gamma, (Q_1, \ldots, Q_{i-1},Q_i',Q_{i+1},\ldots,Q_n)\right)\right).$$
\end{mydef}

In a quantum correlated equilibrium, each player can maximize his expected utility by using his prescribed quantum circuit on his qubits and playing the result, given that all other players follow the output of their circuits. In our definition, $\ket{\psi}$ must be a pure quantum state. We believe that restricting $\ket{\psi}$ to be pure is the natural definition for our purpose. In particular, since our goal is to have a mediator-free setting, allowing for a mixed state would create a fundamental difficulty of how to construct the initial state in a secure manner. See Appendix~\ref{purestates} for a detailed discussion of this restriction.

It is a standard result from quantum computation that, given any quantum circuit $Q_i$, there exists an equivalent quantum circuit $Q_i'$ which performs all measurements at the very end of the computation.\footnote{This is sometimes known as the ``principle of deferred measurement.''} Let $(\ket{\psi}, \Gamma, Q)$ be a QCE of a normal form game. We can assume without loss of generality that every quantum circuit $(Q_1, \ldots, Q_n)$ performs all of its measurements at the end of the circuit's computation. Consider the state $\ket{\psi'}$ which results immediately before any circuit performs a measurement but after all unitary transformations have taken place. (By the ``no-communication theorem'' of quantum mechanics, the final state does not depend on the particular order in which the circuits act, since each circuit acts on separate qubits.) We can assume without loss of generality that the action $a_i \in A_i$ output by $Q_i$ is obtained by measuring the first $\log_2{|A_i|}$ bits of player's $i$ partition of $\ket{\psi'}$ in the standard basis (where we have a canonical mapping between $\log_2{|A_i|}$-bit binary strings and elements of $A_i$.)

\begin{mydef}
	Let $G$ be a normal form game with $n$ players. Let $\ket{\psi}$ be a  pure quantum state, and let $\Gamma$ be a partition of the qubits of $\ket{\psi}$ into $n$ sets $q_1, \ldots, q_n$. For each player $i$, let $A_i$ be the set of actions available to player $i$ in $G$, and fix some mapping between binary strings of length $\log_2|A_i|$ and elements of $A_i$. Let $M_i : q_i \rightarrow A_i$ be the circuit which measures the first $\log_2|A_i|$ qubits of $q_i$ in the standard basis and outputs the resulting action in $A_i$ (using the fixed mapping between strings and actions). If $(\ket{\psi}, \Gamma, (M_1, \ldots, M_n))$ is a QCE, we call  $(\ket{\psi}, \Gamma, (M_1, \ldots, M_n))$ a \textbf{canonical implementation QCE}.
\end{mydef}

From the above discussion, we know that any QCE in a normal form game has an equivalent canonical implementation, by letting $\ket{\psi'}$ be the quantum state which occurs immediately before any measurements occur and after all unitary operations are performed.

\begin{lem}\label{canonical}
Let $G$ be a normal form game, and let $(\ket{\psi}, \Gamma, Q)$ be a QCE. Then there exists a canonical implementation QCE $(\ket{\psi'}, \Gamma', (M_1, \ldots, M_n))$ such that $$D(\ket{\psi}, \Gamma, Q) = D(\ket{\psi'}, \Gamma', (M_1, \ldots, M_n)).$$
\end{lem}

Since every QCE in a normal form game has an equivalent canonical implementation QCE, it follows that any outcome distribution of a QCE in a normal form game can be achieved by a classical correlated equilibrium of the same game. We will see later that the analogous result is false for extensive form games.

\begin{thm}\label{normalformthm}
Let $G$ be a normal form game, and let $(\ket{\psi}, \Gamma, Q)$ be a quantum correlated equilibrium of $G$. Then there exists a classical correlated equilibrium of $G$ which induces the same outcome distribution $D(\ket{\psi}, \Gamma, Q)$.
\end{thm}

\begin{proof}
By Lemma~\ref{canonical}, there exists a canonical implementation QCE $(\ket{\psi'}, \Gamma', (M_1, \ldots, M_n))$ which induces the output distribution $D(\ket{\psi}, \Gamma, Q)$. Let $P$ be the probability distribution on binary strings resulting from measuring all of the qubits of $\ket{\psi'}$ in the standard basis, and consider the classical correlating device which chooses a binary string according to $P$ and tells each player the move suggestion corresponding to the first $\log_2{|A_i|}$ bits of his partition of this binary string. If each player indeed follows the advice, it obviously induces the outcome distribution $D(\ket{\psi}, \Gamma, Q)$. Furthermore, each player can maximize his utility by following his advice: If on the contrary player $i$ could improve his expected utility by not following his suggested move in this classical setting, then we could design a quantum circuit $Q_i$ for $i$ which improves his utility over $u_i(D(\ket{\psi}, \Gamma, Q))$ by deviating in a similar way, thereby violating the assumption that $(\ket{\psi}, \Gamma, Q)$ is a QCE.
\end{proof}

As a simple example of a QCE (similar to an example from \cite{huberman}), consider the normal form game in Figure~\ref{simple2}. The outcome distribution $\frac{1}{2}(TR+BL)$ is achievable by a classical correlated equilibrium. Furthermore, we can achieve this outcome distribution in a QCE by using the entangled state $\frac{1}{\sqrt{2}}(\ket{0}\ket{1} + \ket{1}\ket{0})$ in a canonical QCE representation. In this QCE, the first player measures the first qubit of the pair to determine his move, and the other player measures the second qubit.\footnote{We assume a canonical mapping between binary values and moves in the game, where a measurement value of ``0'' in the first qubit corresponds to the move ``T'', etc.} It is obvious that no player can improve his utility by using a different quantum circuit to manipulate his qubit, since in this outcome distribution each player is always best-responding to the other player's action.

\begin{figure}[htb]
\begin{center}
\begin{game}{2}{2}
	& $L$	& $R$ \\
$T$	&$0,0$	&$1,5$\\
$B$	&$5,1$	&$0,0$
\end{game}
\caption{The distribution $\frac{1}{2}(TR+BL)$ has a canonical QCE with state $\frac{1}{\sqrt{2}}(\ket{0}\ket{1} + \ket{1}\ket{0})$.}
\label{simple2}
\end{center}
\end{figure}

We now ask whether the converse of Theorem~\ref{normalformthm} is true. Consider the game in Figure~\ref{hard}. It is easy to check that $\frac{1}{3}(TR+BL+BR)$ is the outcome of a classical correlated equilibrium. To implement this distribution, we might try having the players share the entangled state $\frac{1}{\sqrt{3}}(\ket{0}\ket{1} + \ket{1}\ket{0} + \ket{1}\ket{1})$ (where the first qubit belongs to the row player, and the second qubit belongs to the column player) and instructing each player to measure his qubit in the standard basis to determine his action. However, this is not a QCE. For example, the row player can apply a Hadamard transformation to his qubit, resulting in the entangled state
$$\frac{2}{\sqrt{6}}\ket{01} + \frac{1}{\sqrt{6}}\ket{00}-\frac{1}{\sqrt{6}}\ket{10}$$
before the measurements. Given that the column player indeed obeys the protocol and simply measures in the standard basis, the resulting outcome distribution is $\frac{2}{3}TR + \frac{1}{6}TL + \frac{1}{6}BL$, which increases the expected utility for the row player.\footnote{While in this example the column player's utility also increases when the row player deviates, if we consider changing the payoff of $TL$ to be $(0,-20)$, then the column player suffers significant losses when the row player deviates in this proposed implementation.} Since the row player can use a Hadamard transformation to increase his expected utility, the state $\frac{1}{\sqrt{3}}(\ket{01} + \ket{10} + \ket{11})$ does not form a canonical QCE implementation of the outcome distribution $\frac{1}{3}(TR+BL+BR)$.

\begin{figure}[htb]
\begin{center}
\begin{game}{2}{2}
	& $L$	& $R$ \\
$T$	&$0,0$	&$6,6$\\
$B$	&$6,6$	&$0,0$
\end{game}
\end{center}
\caption{The CE outcome distribution $\frac{1}{3}(TR+BL+BR)$ cannot be implemented by any QCE. See Appendix~\ref{qceappendix} for a proof of this result.}
\label{hard}
\end{figure}

While the obvious QCE implementation attempt failed, we could conceivably try to design a more complicated QCE protocol achieving this outcome distribution.\footnote{There are examples for which the ``obvious'' approach of achieving a desired distribution fails, but sharing a larger quantum state achieves the distribution in QCE. For example, if we were to change the column player's payoff to always be 0 in the game from Figure~\ref{hard}, we could achieve the $\frac{1}{3}(TR+BL+BR)$ outcome distribution by using an initial 3-qubit shared state, where the last 2 qubits belong to the column player.} In the most technical result of this paper, we prove that there is in fact no QCE achieving this outcome distribution, and thus classical CE is a strictly more powerful concept than QCE in normal form games. The proof of this result is in Appendix~\ref{qceappendix}.

\begin{thm}
There exists a normal form game $G$ and a classical correlated equilibrium distribution of $G$ such that the distribution cannot be achieved by any QCE.
\end{thm}

\section{QCE in Extensive Form Games}

We define a quantum correlated equilibrium in perfect-recall extensive form games analogously to our definition for normal form games. In a QCE, the qubits of a pure quantum state are partitioned and given to the players of a game. A QCE consists of a quantum circuit for each information set. When an information set is reached during the the game, the information set's owner uses all qubits in his possession\footnote{Players have access to an arbitrarily large supply of ancillary $\ket{0}$ qubits.} as input to the appropriate circuit to determine his next action.\footnote{We now care not only about the action output by the circuit, but also about the resulting quantum state, since the player will use this state in later information sets.} In a QCE, no player can improve his expected utility by changing any number of the circuits on his own information sets.

In our definition, the players share the entangled state $\ket{\psi}$ at the start of the game, and do not gain access to any additional entangled qubits as play progresses. This framework is analogous to a classical IR-EFCE. We believe that our definition is natural, since it avoids the necessity of a mechanism to distribute new entangled states in later information sets.

We show in Appendix~\ref{comparison} that every IR-EFCE in an extensive form game $G$ has a corresponding classical CE in the normal form equivalent $n(G)$. Furthermore, Theorem~\ref{normalformthm} states that every QCE of $n(G)$ has an equivalent classical CE in $n(G)$. Nevertheless, it is possible that $G$ has outcome distributions which can be achieved by a QCE but which cannot be achieved by any classical IR-EFCE (or EFCE).

The underlying reason why quantum correlated equilibrium can be more powerful in an extensive form game $G$ than in $n(G)$ is the measurement principle of quantum mechanics. In particular, an action in $n(G)$ specifies a choice of action for \textit{every} information set of $G$, even those information sets which are not reached in the actual execution. To specify our actions in all of these information sets for the game $n(G)$, we would need to operate on $\ket{\psi}$ many times to determine what we would hypothetically do in all of these unreached information sets. In the extensive form game, a player only operates on $\ket{\psi}$ when his information set is actually reached.\footnote{Because the operations performed to $\ket{\psi}$ depend on the information sets visited during execution of the game, we do not have a concept analogous to a ``canonical implementation'' of an extensive form QCE.}

\subsection{The Complete-Information GHZ Game}

The analog of Theorem~\ref{normalformthm} is false for extensive form games. In particular, we have the following result.

\begin{thm}\label{ghztheorem}
There exists a complete information extensive form game $G$ and an outcome distribution of the game which can be achieved by a QCE but not by any EFCE (or by any IR-EFCE).
\end{thm}

We present a full proof of this theorem in Appendix~\ref{ghz}. The game that we use is a slight modification of the GHZ game from \cite{ghz}. The GHZ game is a well-known example of a scenario where players can achieve higher utility in a quantum setting than they can achieve classically. While the GHZ game is a three-player game of incomplete information, we construct a ``complete-information GHZ game'' (denoted cGHZ) by introducing a fourth player, who we incentivize to act as ``nature.'' An ``always succeed'' outcome distribution (where the original three players always receive maximum payoff and the nature player receives minimum payoff) can be achieved in a QCE but not in any EFCE or IR-EFCE.

By combining Theorem~\ref{normalformthm} and Theorem~\ref{ghztheorem} we obtain the immediate corollary that extensive form QCE is a more powerful concept than normal form QCE.

\begin{cor}
There exists extensive-form game $G$ and an outcome distribution of the game which can be achieved by a QCE, yet no corresponding outcome distribution can be achieved by a QCE in the normal form equivalent game $n(G)$. 
\end{cor}

\subsection{Limits of QCE in Extensive Form Games}

While in some extensive form games QCE can be a more powerful solution concept than classical EFCE, there are other games where outcome distributions can be achieved by EFCE but cannot be implemented by any extensive form QCE. Consider the game in Figure~\ref{EFCEnoQCE}. As discussed in Appendix~\ref{comparison}, the outcome distribution  $1/2(IN,a,L) + 1/2(IN,b,R)$ can be achieved by an EFCE but not by any IR-EFCE. For a nearly identical reason, this distribution cannot be achieved by any extensive form QCE.

\begin{figure}[htb]
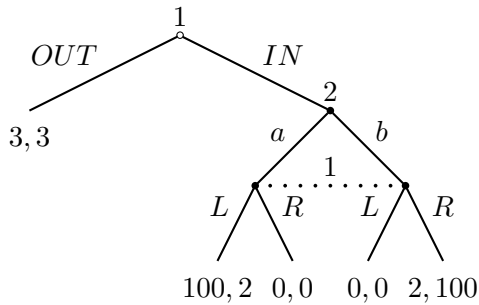
 
\hspace*{\fill} 
\begin{egame}(500,380) 
\putbranch(150,340)(2,1){200} \iib{1}{$OUT$}{$IN$}[$3,3$][] 
\putbranch(350,240)(1,1){100} \iib{2}{$a$}{$b$}
\infoset(250,140){200}{1}
\putbranch(250,140)(1,2){50}\iib{}{$L$}{$R$}[$100,2$][$0,0$]
\putbranch(450,140)(1,2){50}\iib{}{$L$}{$R$}[$0,0$][$2,100$]

\end{egame} 
\hspace*{\fill} \caption[]{The outcome distribution $1/2(IN,a,L) + 1/2(IN,b,R)$ can be achieved by an EFCE but not by any QCE (or by any IR-EFCE).}\label{mygame} 
\label{EFCEnoQCE}
\end{figure} 

Suppose on the contrary that there were some QCE achieving the outcome distribution  $1/2(IN,a,L) + 1/2(IN,b,R)$. Then we notice that, at the beginning of the game, player 1 could simulate the quantum circuit for his second information set to compute whether, if he were to play IN, his next advice would be $a$ or $b$. If he computes that his next advice will be $b$, then he can improve his utility by deviating and playing OUT.

The underlying reason why the outcome distribution discussed above cannot be implemented by a QCE is that, in a QCE, a player has the ability to apply his circuits early, and can thereby compute what his advice would be if he were to reach certain later information sets in the future. Since in this example player 1 would have no further need of his qubits if he were to play OUT (since the game would end immediately), there is no penalty for him to discover what his future advice will be. We have therefore proven the following theorem:

\begin{thm}
There exists an extensive form game $G$ and an outcome distribution of the game which can be achieved by an EFCE but not by any QCE or by any IR-EFCE.
\end{thm}

While the above theorem states that in some games the EFCE concept can be more powerful than both QCE and IR-EFCE, it also can be the case that some distributions can be implemented by both an EFCE and by a QCE but not by any IR-EFCE. By combining aspects of the game from Figure~\ref{EFCEnoQCE} with complete-information GHZ game, we have the following result, which we prove in Appendix~\ref{finalproof}.

\begin{thm}\label{quantumefce}
There exists an extensive form game $G$ and an outcome distribution of the game which can be achieved by an EFCE and by a QCE but not by any IR-EFCE.
\end{thm}

Finally, we note that the normal form game from Appendix~\ref{qceappendix} (viewed as a depth-2 imperfect information extensive form game) provides an example of an outcome distribution that can be implemented by EFCE and by IR-EFCE but not by any QCE.


\begin{figure}[htb]\label{summarydiagram}
\begin{center}
\includegraphics[scale=0.35]{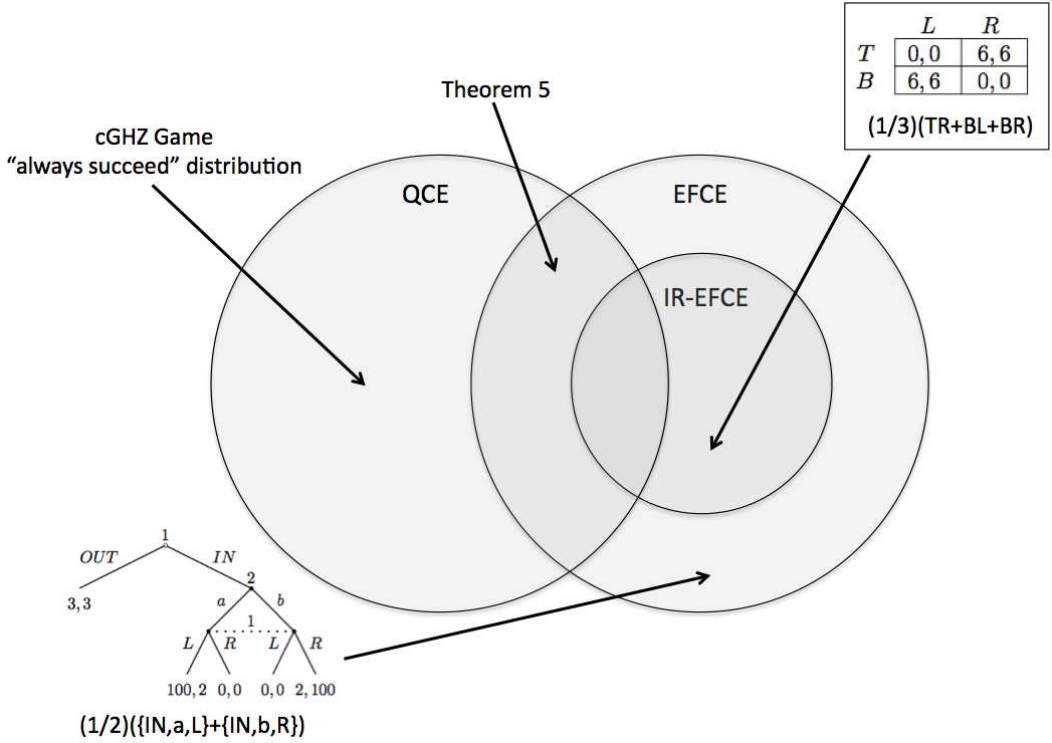}
\caption{Examples of outcome distributions implementable by QCE, EFCE, and IR-EFCE}
\end{center}
\end{figure}

\section{Further Work}

\subsection{Perfect Quantum Implementation of Classical CE}

A potential application of QCE is to use quantum entanglement (and no classical communication) to remove the need for a trusted mediator when implementing a classical correlated equilibrium. While this might not always be possible (since some CE distributions might not have a corresponding QCE- see Appendix~\ref{qceappendix}), we have shown that for many classical CE distributions there indeed exists a QCE which induces the same distribution.

We now ask more precisely what it means for a mechanism to ``implement'' a classical CE. If we were to follow the viewpoint of Dodis, Halevi and Rabin \cite{dodis}, it would suffice to show that our mechanism has an equilibrium which is equivalent to the distribution of the desired correlated equilibrium. The framework of \cite{dodis} matches closely with the analysis we have already performed, since we studied which correlated equilibria have a corresponding QCE with the same outcome distribution.

We can also take a more restrictive view of what it means for a mechanism to ``implement'' a desired correlated equilibrium, in a manner analogous to Izmalkov, Lepinski and Micali \cite{ilm1},  \cite{ilm2}. Izmalkov, Lepinski, and Micali defined the concept of ``perfect implementation'' of a mechanism. Roughly speaking, a perfect implementation preserves not only a single desired equilibrium, but it must preserve all of the strategic properties of the game as well as the privacy of the players. We do not wish to formally define perfect quantum implementation at this time, but we will give a rough outline of some of the properties it should obey.

For simplicity, we will only look at normal-form classical games, and we will continue to view the underlying game as a ``black box'' (avoiding the implementation issues from \cite{ilm1}). Let $D$ be a classical CE distribution of the game $G$, and let $A$ be the trusted classical mediator which suggests actions to the players according to $D$. We denote $G^A$ to be the classical game where each player receives advice from $A$ before deciding on his action. While we do not wish to formalize the notion at this point, we draw motivation from \cite{ilm2}, and impose the requirement that, in order for $(\ket{\psi},\Gamma, \tilde{Q})$ to be a \textit{perfect quantum implementation} of $D$, we must at the very least satisfy
\begin{itemize}
	\item $(\ket{\psi},\Gamma, \tilde{Q})$  is a QCE with outcome distribution $D$.
	\item For each player $i$ there is a mapping $f_i$ from quantum circuits $Q_i$ to strategies in $G^A$ such that, if $(\ket{\psi}, \Gamma, (Q_1, \ldots, Q_n))$ is a QCE with distribution $D'$, then the set of strategies $(f_1(Q_1), \ldots, f_n(Q_n))$ is an equilibrium of $G^A$ with outcome distribution $D'$.
\end{itemize}
We should also enhance our definition to require that (if the game has more than two players) properties such as collusion-resilience of $G^A$ are preserved in our perfect quantum implementation. For example, if two players in $(\ket{\psi},\Gamma, \tilde{Q})$ could collude (perhaps by performing a quantum operation which acts on both of their qubits) then a similar collusion should be possible in $G^A$.

The main idea is that a perfect quantum implementation not only achieves $D$ in QCE, but does not introduce any additional equilibria which would not already exist if the players were given the classical mediator $A$.\footnote{Since we do not introduce any communication between players in our quantum setting, we do not have to deal with issues of ``aborting'' computations as in \cite{ilm1} and \cite{ilm2}.}

Achieving a perfect quantum implementation of a classical CE seems to be a much loftier goal than matching a single desired distribution in equilibrium, and we suspect that in most cases will be impossible. While in some very simple examples we are able to achieve such a perfect implementation, it remains a further question to study to what extent we can achieve a perfect (or some reasonably-defined approximation of perfect) quantum implementation of a classical CE.

\subsection{Other Open Questions}

\begin{enumerate}
	\item What is the computational complexity of computing a QCE in a normal form game? In an extensive form game?
	\item Given an outcome distribution of a classical game which can be achieved by a QCE, is there an efficient method of computing the smallest number of entangled qubits that must be shared in order to achieve this outcome distribution in QCE? If the game is a normal form game, is there an efficient method of determining the smallest  number of qubits needed in a shared state which achieves the outcome distribution in a canonical implementation QCE?
	\item In our model for QCE, the players are allowed to initially share an arbitrary pure quantum state. Which QCE outcome distributions are possible if we only allow the players to initially share an arbitrary number of EPR pairs?\footnote{Using EPR pairs, it is possible for two players to construct an arbitrary shared entangled state using only local operations and classical communication. (See \cite{bennett}.) However, this approach is unsatisfactory in our framework for a two-player game, since it requires classical communication (which changes the underlying classical game) and this protocol is not robust against a dishonest opponent. Notice that if a state $\ket{\psi}$ can be constructed (in an honest setting) by only local operations (with no communication) starting from shared EPR pairs, and if $\ket{\psi}$ can be used in a QCE, then the overall construction will be a QCE, regardless of the robustness of the state-constructing protocol. This follows from reversibility properties of quantum mechanics.}
	\end{enumerate}

\section{Acknowledgements}

I would like to thank Scott Aaronson and Silvio Micali for their advice and for many helpful discussions.

\appendix

\section{Normal Form Correlated Equilibrium}\label{normalform}

In a Nash equilibrium of a normal form game, each player selects a probability distribution over his possible moves, and the resulting distribution on outcomes is the resulting product distribution. In a correlated equilibrium, however, a \textit{correlating device} is used to correlate the random choices made by each player, therefore allowing for a wider variety of outcome distributions.

We consider a canonical representation of correlated equilibria, in which a correlating device suggests a single move to each player. (In a more general framework, the correlating device can provide arbitrary signals, although these two models are equivalent.) The resulting play is a CE if it is optimal for each player to always follow his advice, given that all other players follow their advice.

For example, consider the game in Figure~\ref{game1}. We claim that there exists a CE having outcome distribution $\frac{1}{3}(TR + BL + BR)$. Imagine the correlating device taking three envelopes with contents $TR$, $BL$, and $BR$, choosing an envelope at random, and secretly telling each player his recommended move. For this to be a correlated equilibrium, we must show that each player maximizes his expected utility by always following his advice, given that his opponents always follow their advice. For example, suppose that the column player always follows his advice. If the row player receives advice $T$, then he knows that the column player must have received advice $R$, and therefore the row player maximizes his utility by following his suggestion and playing $T$. If the row player instead receives advice $B$, then he knows that, conditional on his advice, his opponent will be playing $L$ half the time and $R$ half the time. His expected utility of playing $T$ is therefore 3.5, while his expected utility of $B$ is 5. Therefore the row player maximizes his expected utility by playing $B$. An analogous argument shows that, if the row player always follows his advice, then it is optimal for the column player to follow his advice as well.

\begin{figure}[htb]
\begin{center}
\begin{game}{2}{2}
	& $L$	& $R$ \\
$T$	&$0,0$	&$7,10$\\
$B$	&$10,7$	&$0,0$
\end{game}
\end{center}
\caption{The outcome distribution $\frac{1}{3}(TR + BL + BR)$ can be achieved by a classical correlated equilbrium.}
\label{game1}
\end{figure}

Notice that the presence of a trusted correlated device is vital to achieve the equilibrium distribution $\frac{1}{3}(TR + BL + BR)$ in Figure~\ref{game1}. In particular, we cannot achieve this distribution using only a public random string. The underlying reason is that each player must not be able to know the opponent's advice. If, for example, the row player received advice $B$ and could compute with certainty whether the column player had received advice $L$ or $R$, then he could increase his expected utility by playing $T$ whenever the column player received $R$.





\section{Comparison of Classical Correlated Equilibria}\label{comparison}

We state a few results concerning the power of various classical correlated equilibrium concepts. The arguments given below are informal. When we compare equilibrium concepts, we are concerned primarily with the \textit{induced outcome distribution} reached in each equilibrium.\footnote{Thus, an informal statement such as ``an IR-EFCE has an equivalent EFCE'' should be interpreted as meaning that there is an EFCE having the same induced outcome distribution as the IR-EFCE.}

\begin{claim}
In any extensive form game, any IR-EFCE has an equivalent EFCE. However, it is not necessarily the case that every EFCE has an equivalent IR-EFCE.
\end{claim}

\begin{proof}

We argue informally that the above claim is true by noting that at any information set in an IR-EFCE, every player knows at least as much as he would know in the corresponding EFCE, and he also knows additional information about the advice he will receive at later information sets. Therefore, if it is never in the player's interest to deviate even if he knows all of the information from the IR-EFCE, it will clearly be impossible for him to increase his expected utility by deviating if he knows even less about the future.

For the reverse direction of the claim, consider the game in Figure~\ref{EFCEnoQCE}. It is easy to check that the outcome distribution $1/2(IN,a,L) + 1/2(IN,b,R)$ is achievable by an EFCE. However, there is no IR-EFCE achieving this distribution. The reason is that, in the IR-EFCE framework, if player 1 receives advice $IN$, he can check immediately whether his advice for his second move will be $L$ or $R$. If his advice will be $R$, he can improve his utility by deviating and playing $OUT$. Notice the EFCE framework, the player does not know if his advice will be $L$ or $R$ until \textit{after} he plays $IN$, and therefore he simply computes that his expected utility of playing $IN$ is 51 which is greater than the utility of 3 he receives by playing $OUT$.
\end{proof}

For any extensive form game $G$, there is a corresponding normal form game, which we will denote $ n(G)$ and call the \textit{normal form equivalent} of $G$. The players of $n(G)$ are the same as the players of $G$, and a pure strategy of player $i$ in $n(G)$ corresponds to a choice of a single move from \textit{all} of $i$'s information sets in $G$. (Note that the size of the game matrix for $n(G)$ might be exponentially larger than the size of the game tree representation of $G$.) We define the payoffs of $n(G)$ according to the corresponding outcome of $G$. We say that an outcome distribution $D$ of $G$ \textit{corresponds} to an outcome distribution $D'$ of $n(G)$ if the probability of any outcome $x$ of $G$ under $D$ is equal to the sum of the probabilities of all corresponding outcomes $x_1, x_2, \ldots$ of $n(G)$ under $D$. (Notice that $x$ may have several corresponding outcomes in $n(G)$, since many moves in $n(G)$ can differ only on information sets which are never reached in the path to $x$ in $G$'s game tree.) We say that equilibria of $G$ and $n(G)$ are \textit{equivalent} if their induced outcome distributions are equivalent.

\begin{claim}\label{IREFCE=CE}
Let $G$ be an extensive form game. Then every IR-EFCE of $G$ has an equivalent correlated equilibrium of $n(G)$. Furthermore, every correlated equilibrium of $n(G)$ has an equivalent IR-EFCE of $G$.
\end{claim}
The proof of Claim~\ref{IREFCE=CE} follows easily from the definition of $n(G)$.

\section{Pure States in Quantum Correlated Equilibria}\label{purestates}


In our definition of QCE, the players initially share a pure quantum state $\ket{\psi}$. We do not allow for the players to initially share a mixed state (which would make the converse of Theorem~\ref{normalformthm} obviously true), since the use of mixed states hides information and affects properties such as collusion-resiliance. For example, suppose that, in order to implement the distribution $\frac{1}{3}(TR+BL+BR)$ in Figure~\ref{appgame}, we allowed the players to share the mixed state $\rho$ which was $\ket{01}$ with probability 1/3, $\ket{10}$ with probability 1/3, and $\ket{11}$ with probability 1/3. The access to such a mixed state obviously allows the players to achieve the desired outcome distribution (by measuring in the standard basis), and it is furthermore obvious that no player can improve his utility by deviating. (Since our mixed state has only classical uncertainty of which state the players share, the players cannot take advantage of quantum interference effects.)

The difficulty comes when we ask how we obtained the mixed state $\rho$. One method is to begin with the state
$$\ket{\psi} = \frac{1}{\sqrt{3}}(\ket{0}\ket{1}\ket{00} + \ket{1}\ket{0}\ket{01} + \ket{1}\ket{1}\ket{10})$$
and to give the first qubit from $\ket{\psi}$ to player 1 and the second qubit to player 2. The complication is what we do with the remaining two qubits. If we were to give these qubits to either player 1 or player 2, the player could measure the qubits to deduce the other player's advice, and therefore has the opportunity to deviate and improve his utility. An implementation which relied on the mixed state $\rho$ would therefore be undesirable in that, if either player managed to get access to the missing two qubits of the purification, he might be able to improve his utility by deviating.

One way to solve this problem is to introduce a third player to the game, to give this player only a single possible action $O$ in $G$, and to give him utility 0 in all outcomes of the game. Classically, this three-player game inherits much of the equilibrium structure of the original game, since the third player has no choice in his action. Furthermore, we can indeed achieve the outcome distribution $\frac{1}{3}(TRO+BLO+BRO)$ by a QCE in this three-player game by using the state $\ket{\psi}$ from above and giving the last two qubits to the third player.\footnote{This construction works in general- given any classical CE outcome distribution of a normal form game $G$, we can construct a game $G'$ which has one additional player (where this new player only has a single action, and receives utility 0 regardless of the outcome) such that we can implement the corresponding outcome distribution in $G'$ by a QCE.}

From a mechanism-design point of view, however, this QCE in the three-player game is very different from the classical CE implementation in the two-player game. In particular, the addition of the third player makes the three-player quantum implementation vulnerable to collusion (since either player 1 or player 2 could improve his utility by colluding with player 3) while the two-player classical implementation does not have this vulnerability. Since we do not want to introduce unnecessary collusive vulnerabilities into our implementation, we find this construction of adding an additional player to the game unsatisfactory. Indeed, our main motivation for introducing QCE was to avoid the necessity of a trusted third party, and allowing for mixed states reintroduces this difficulty.

The underlying difficulty of using mixed states is that, in their construction, information leaks outside of the system. If a player were to gain access to this information, then he might use this knowledge to improve his utility. It is therefore difficult to imagine a method of constructing an initial shared mixed state without use of a third party. By using only pure states, we can avoid much of the difficulty in preparing the initial state.

\section{Normal Form CE with no Quantum Equivalent}\label{qceappendix}

In this Appendix, we study a particular CE in a normal form game and prove that no QCE achieves this outcome distribution. The game we analyze is as follows: 

\begin{figure}[htb]
\begin{center}
\begin{game}{2}{2}
	& $L$	& $R$ \\
$T$	&$0,0$	&$6,6$\\
$B$	&$6,6$	&$0,0$
\end{game}
\end{center}
\caption{The CE distribution $\frac{1}{3}(TR+BL+BR)$ cannot be achieved in a QCE.}
\label{appgame}
\end{figure}
This game has a classical correlated equilibrium $\frac{1}{3}(TR+BL+BR)$ with expected utility $4$ for each player. As shown earlier, this outcome distribution is not implemented in canonical QCE by the shared state $\frac{1}{\sqrt{3}}(\ket{01}+\ket{10}+\ket{11})$, since either player can measure in the Hadamard basis to improve his utility.

In this Appendix, we prove that it is impossible to use a more complicated shared state to obtain a canonical QCE implementation of the desired distribution.\footnote{Recall that if this distribution can be implemented by a QCE, then it has a canonical QCE implementation, so it suffices to prove that there is no canonical QCE.} In Appendix~\ref{failed}, to introduce of our analysis technique, we present an incorrect implementation attempt, and prove that one player has incentive to deviate from his prescribed protocol. In Appendix~\ref{generalization} we generalize this analysis, resulting in a set of conditions that any shared state implementing the outcome distribution in a canonical QCE must satisfy. Finally, in Appendix~\ref{impossibility}, we prove that no state satisfies these conditions. 

\subsection{Example Analysis of Failed Implementation Attempt}\label{failed}

To give intuition for our analysis approach, we present a failed attempt at achieving the desired outcome distribution in a canonical QCE. 
Consider having the players share the state

\begin{align*}
\ket{\psi} &= \frac{1}{\sqrt{6}}\ket{00}\ket{10} - \frac{1}{\sqrt{6}}\ket{00}\ket{11} + \frac{1}{\sqrt{6}}\ket{10}\ket{00} + \frac{1}{\sqrt{6}}\ket{11}\ket{00} + \frac{1}{\sqrt{12}}\ket{10}\ket{10} \\
& + \frac{1}{\sqrt{12}}\ket{10}\ket{11} + \frac{1}{\sqrt{12}}\ket{11}\ket{10} + \frac{1}{\sqrt{12}}\ket{11}\ket{11}
\end{align*}
where the first two qubits belong to the row player and the last two qubits belong to the column player. A player is instructed to measure his qubits in the standard basis, and to play the action according to his first qubit. (Thus, if the row player measures a $\ket{0}$ in the first register, he should play $T$, etc.)

Suppose that player 2 (the column player) follows his prescribed protocol. Player 1 (the row player) is now facing the following scenario:

\begin{center}
\begin{tabular}{c|c|c|c|c}
\hline
Probability & P2's state & P1's state  & $u(T)$ & $u(B)$ \\
\hline
1/3 & $\ket{00}$ & $\frac{1}{\sqrt{2}}\ket{10}+\frac{1}{\sqrt{2}}\ket{11}$ & 0 & 6 \\ 
1/3 & $\ket{10}$ & $\frac{1}{\sqrt{2}}\ket{00} + \frac{1}{2}\ket{10} + \frac{1}{2}\ket{11}$  & 6 & 0 \\
1/3 & $\ket{11}$ & $- \frac{1}{\sqrt{2}}\ket{00} + \frac{1}{2}\ket{10} + \frac{1}{2}\ket{11}$  & 6 & 0\\
\end{tabular}
\end{center}

We now write $\ket{1+} = \frac{1}{\sqrt{2}}(\ket{10}+\ket{11})$. With this notation, the row player's situation is as follows:
\begin{center}
\begin{tabular}{c|c|c|c}
\hline
Probability & P1's state  & $u(T)$ & $u(B)$ \\
\hline
1/3 & $\ket{1+} $& 0 & 6 \\ 
1/3 &  $\frac{1}{\sqrt{2}}\ket{1+} + \frac{1}{\sqrt{2}}\ket{00}$  & 6 & 0 \\
1/3 &  $\frac{1}{\sqrt{2}}\ket{1+} - \frac{1}{\sqrt{2}}\ket{00}$  & 6 & 0\\
\end{tabular}
\end{center}
Since $\ket{1+}$ and $\ket{00}$ are orthogonal, we observe that player 1's performance of any strategy for the above scenario would perform identically to using that strategy in the scenario:
\begin{center}
\begin{tabular}{c|c|c|c}
\hline
Probability & P1's state  & $u(T)$ & $u(B)$ \\
\hline
1/3 & $\ket{1+} $& 0 & 6 \\ 
1/3 &  $\ket{1+}$  & 6 & 0 \\
1/3 &  $ \ket{00}$  & 6 & 0\\
\end{tabular}
\end{center}

The above scenario is simply a classical probabilistic mixture of the orthogonal advice states $\ket{1+}$ or $\ket{00}$. Therefore, player 1 can do no better than playing $B$ when given $\ket{1+}$ and playing $T$ when given $\ket{00}$, for an expected utility of 4. This is the expected utility of following the QCE protocol, and therefore player 1 cannot gain any utility by deviating.

Unfortunately, the state $\ket{\psi}$ does not yield a canonical QCE. The reason is that the column player can improve his utility by deviating. By performing similar analysis to that above, we see that the column player is facing the following scenario:

\begin{center}
\begin{tabular}{c|c|c|c|c}
\hline
Probability & P1's state & P2's state  & $u(L)$ & $u(R)$ \\
\hline
1/3 & $\ket{00}$ & $\frac{1}{\sqrt{2}}\ket{10}-\frac{1}{\sqrt{2}}\ket{11}$ & 0 & 6 \\ 
1/3 & $\ket{10}$ & $\frac{1}{\sqrt{2}}\ket{00} + \frac{1}{2}\ket{10} + \frac{1}{2}\ket{11}$  & 6 & 0 \\
1/3 & $\ket{11}$ & $\frac{1}{\sqrt{2}}\ket{00} + \frac{1}{2}\ket{10} + \frac{1}{2}\ket{11}$  & 6 & 0\\
\end{tabular}
\end{center}
We observe that the states $\frac{1}{\sqrt{2}}\ket{10}-\frac{1}{\sqrt{2}}\ket{11}$ and $\frac{1}{\sqrt{2}}\ket{00} + \frac{1}{2}\ket{10} + \frac{1}{2}\ket{11}$ are orthogonal. Therefore, the column player can distinguish the first case from the bottom two cases with certainty, and can therefore achieve expected utility $\frac{1}{3}(6+6+6)=6$. Thus, he has incentive to deviate from his prescribed protocol.

\subsection{Density Matrix Constraints}\label{generalization}

In this section, we generalize the analysis from Appendix~\ref{failed} to study the properties of any state $\ket{\psi}$ which implements the outcome distribution $\frac{1}{3}(TR+BL+BR)$ in a canonical QCE. We will prove in Appendix~\ref{impossibility} that no such state $\ket{\psi}$ exists.

We first look at the mixed state of the row player conditional on the column player's advice. Conditional on the column player having first qubit $\ket{0}$, we call the density matrix of the row player's qubits $\rho$, and conditional on the column player's first qubit being $\ket{1}$, we call the row player's density matrix $\sigma$.

We know that when the column player has first qubit $\ket{0}$, the first qubit of the row player's state must be $\ket{1}$ (since $TL$ is never played in the equilibrium), while when the column player has first qubit $\ket{1}$, the row player must have equal probability of measuring $\ket{0}$ or $\ket{1}$ in his first qubit. Therefore, we can write

$$\rho = \left[ \begin{array}{c|c} 0 & 0 \\ \hline 0 & \tilde{\rho} \end{array} \right] ;  \qquad
 \sigma = \left[ \begin{array}{c|c} \sigma_1 & \sigma_2 \\ \hline \sigma_2^\dagger & \sigma_3 \end{array} \right]$$
where the first half of the diagonal entries (those which lie in the top-left quadrant) correspond to the row player's first qubit being $\ket{0}$ and the second half of the diagonal entries (those in the bottom-right quadrant) correspond to the row player's first qubit being $\ket{1}$. Since $\ket{\psi}$ implements the distribution $\frac{1}{3}(TR+BL+BR)$ in canonical QCE, we have $Tr(\tilde{\rho}) = 1$ and $Tr(\sigma_1)=Tr(\sigma_3) = 1/2$.

The goal of this section is to prove the following lemma:

\begin{lem}\label{densitycondition}
If $\ket{\psi}$ yields a canonical QCE implementation, then $\sigma_2$ is the zero matrix and $\sigma_3 = \frac{1}{2}\tilde{\rho}$.
\end{lem}

Intuitively, the row player maximizes his utility by trying to determine if he's been given $\sigma$ instead of $\rho$. We have the following definition:

\begin{mydef*}
Let $Q$ be a quantum circuit which takes as input a mixed state and outputs either 0 or 1. We define the \textit{completeness} of $Q$, denoted $c(Q)$, to be
$$c(Q) = Pr[Q(\sigma) = 1]$$
and the \textit{soundness} of $Q$, denoted $s(Q)$, to be
$$s(Q) = Pr[Q(\rho) = 1].$$
\end{mydef*}
Notice that by measuring the first qubit and outputting 1 if this qubit is $\ket{0}$, the row player has a simple test which has completeness 1/2 and soundness 0. 
\begin{claim}\label{density}
The row player has incentive to deviate from his QCE protocol if and only if there exists a quantum circuit $Q$ such that
$$c(Q) > \frac{1}{2} + \frac{1}{2} s(Q).$$
\end{claim}
\begin{proof}
Suppose we have such a circuit $Q$. Consider the protocol for the row player where he runs $Q$ on his state, and plays $T$ if and only if $Q$ outputs $1$. The row player's expected utility of this strategy is
$$\frac{1}{3}(6*Pr[Q(\rho) = 0]) + \frac{2}{3}(6*Pr[Q(\sigma)=1])$$
$$= 2(1-s(Q)) + 4c(Q) = 2 + 2(2c(Q) - s(Q))$$
$$ > 2 + 2(1 - s(Q) + s(Q)) = 4$$
and therefore this strategy gives the row player expected utility above 4. Since his expected utility in $\frac{1}{3}(TR+BL+BR)$ is 4, the row player has incentive to deviate.

\bigskip

To show the other direction, suppose that the row player had some circuit $Q'$ which took his qubits as input and output either $T$ or $B$, such that his utility of playing this strategy was strictly greater than 5. Construct the quantum circuit $Q$ which simulates $Q'$ and outputs 1 whenever $Q'$ outputs $T$. We now compute
$$\mathbb{E}[u_{row}(Q')] = \frac{1}{3}(6*Pr[Q'(\rho) = B]) + \frac{2}{3}(6*Pr[Q'(\sigma)=T])$$
$$ = \frac{1}{3}(6*Pr[Q(\rho) = 0]) + \frac{2}{3}(6*Pr[Q(\sigma)=1])$$
$$= 2(1-s(Q)) + 4c(Q) > 4$$
and thus
$$4c(Q) - 2s(Q) > 2$$
so $c(Q) > \frac{1}{2} + \frac{s(Q)}{2}$, as desired.
\end{proof}

\bigskip

As an example application of this claim, we notice that in the attempted QCE implementation $$\frac{1}{\sqrt{3}}(\ket{01}+\ket{10}+\ket{11}),$$ the row player's task is to distinguish between the state $\ket{1}$ (corresponding to $\rho$) and $\frac{1}{\sqrt{2}}(\ket{0}+\ket{1})$ (corresponding to $\sigma$). By measuring in the $\ket{+}, \ket{-}$ basis and outputting 1 if and only if the measurement returns $\ket{+}$, we have a test with completeness 1 and soundness 1/2. Since $1 > .5 + .25$, the row player can indeed improve his utility by deviating.

We now suppose that $\ket{\psi}$ indeed implements a QCE, and we will look at the states
$$\rho = \left[ \begin{array}{c|c} 0 & 0 \\ \hline 0 & \tilde{\rho} \end{array} \right] ;  \qquad
 \sigma = \left[ \begin{array}{c|c} \sigma_1 & \sigma_2 \\ \hline \sigma_2^\dagger & \sigma_3 \end{array} \right]$$
 in detail, by applying Claim~\ref{density}. We first claim that $\sigma_3 = \frac{1}{2}\tilde{\rho}$. In particular, if $\sigma_3 \neq \frac{1}{2}\tilde{\rho}$, then we could perform the following test $Q$:
 \begin{itemize}
 	\item Measure the first qubit in the standard basis. If this qubit is $\ket{0}$, output 1 and halt. If this qubit is $\ket{1}$, continue.
	\item Since $2\sigma_3 \neq \tilde{\rho}$, there exists a quantum circuit $Q'$ acting on the remaining qubits such that
	$$P[Q'(2\sigma_3) = 1] > P[Q'(\tilde{\rho}) = 1].$$
	Simulate this circuit on the remaining qubits, and output the result.
 \end{itemize}
 We notice that if the input state were $\sigma$ and the result of measuring the first qubit was 1, then the resulting state (ignoring the first qubit) would indeed be $2\sigma_3$, Therefore, the test $Q$ has
 $$c(Q) = \frac{1}{2} + \frac{1}{2}Pr[Q'(2\sigma_3 = 1)]$$
 $$s(Q) = Pr[Q'(\tilde{\rho}) = 1]$$
 and hence $c(Q) > \frac{1}{2} + \frac{s(Q)}{2}$. By Claim~\ref{density}, the row player has incentive to deviate.

Our next goal is to show that $\sigma_2$ must be the 0 matrix. To do this, we consider the experiment where we are given mixed state $\sigma$ with probability $\eta_1 = 2/3$ and state $\rho$ with probability $\eta_2 = 1/3$. We have a quantum measurement circuit $Q$ which is supposed to output $1$ when given $\sigma$ and 0 when given $\rho$. We define the error probability of $Q$ to be
$$P_E(Q) = \eta_1 P[Q(\sigma) = 0] + \eta_2 P[Q(\rho) = 1]$$
$$P_E(Q) = \frac{2}{3}(1 - c(Q)) + \frac{1}{3}s(Q) = \frac{2}{3} - \frac{2}{3}c(Q) + \frac{1}{3}s(Q)$$
$$P_E(Q) = \frac{1}{3} + \frac{2}{3}\left(\frac{s(Q)}{2} + \frac{1}{2} - c(Q)\right)$$
We define the minimum error probability to be
$$P^{min}_{E} = \min_QP_E(Q) = \frac{1}{3} - \frac{2}{3} \max_Q \left( c(Q) - \frac{s(Q)}{2} - \frac{1}{2} \right).$$
We now use a result of Helstrom \cite{helstrom}, as mentioned in \cite{herzog}, which states that
$$P^{min}_{E} = \frac{1}{2}(1 - Tr|\eta_2 \rho - \eta_1 \sigma|),$$
where $|\tau| = \sqrt{\tau^\dagger \tau}$. (Thus, the trace of $|\tau|$ is the sum of the singular  values of $\tau$.) By applying Claim~\ref{density}, we know that the row player has incentive to deviate if and only if
$$\frac{1}{2}(1 - Tr|\frac{1}{3}\rho - \frac{2}{3} \sigma|) < \frac{1}{3},$$
or, equivalently, if
$$Tr|\frac{1}{3}\rho - \frac{2}{3}\sigma| > \frac{1}{3}.$$
We now notice that
$$Tr|\frac{1}{3}\rho - \frac{2}{3}\sigma| = Tr|\frac{2}{3}\sigma - \frac{1}{3}\rho| = Tr\left|  \left[ \begin{array}{c|c} \frac{2}{3}\sigma_1 & \frac{2}{3}\sigma_2 \\ \hline \frac{2}{3}\sigma_2^\dagger & 0 \end{array} \right]    \right|$$
Therefore, the row player has incentive to deviate if an only if
$$Tr\left|  \left[ \begin{array}{c|c} \sigma_1 & \sigma_2 \\ \hline \sigma_2^\dagger & 0 \end{array} \right]    \right| > \frac{1}{2}.$$
Since the trace of the absolute value of a matrix equals the sum of its singular values, we can write the above condition as
$$\sum_i \left| \lambda_i\left( \left[ \begin{array}{c|c} \sigma_1 & \sigma_2 \\ \hline \sigma_2^\dagger & 0 \end{array} \right]  \right) \right| > \frac{1}{2}. $$
Since we know that
$$\sum_i \lambda_i\left( \left[ \begin{array}{c|c} \sigma_1 & \sigma_2 \\ \hline \sigma_2^\dagger & 0 \end{array} \right]  \right)   = Tr(\sigma_1) = \frac{1}{2},$$
the row player has incentive to deviate if and only if the matrix $M = \left[ \begin{array}{c|c} \sigma_1 & \sigma_2 \\ \hline \sigma_2^\dagger & 0 \end{array} \right]$ has a negative eigenvalue. We claim that this occurs if and only if $\sigma_2$ is non-zero.

If $\sigma_2 = 0$, then clearly the nonzero eigenvalues of $M$  are equal to the nonzero eigenvalues of $\sigma_1$ (a positive semi-definite matrix). Therefore, if $\sigma_2 = 0$, the row player has no incentive to deviate.

Suppose now that $\sigma_2$ is non-zero. Then there exists some index $(i,j)$ into the top-right block of $M$ such that $M_{(i,j)}$ is non-zero. It follows from our indexing that $(i,i)$ is an index into the top-left block of $M$, $(j,i)$ is an index into the bottom-left block, and $(j,j)$ is an index into the bottom-right block.

To show that $M$ has a negative eigenvalue, we will show that $M$ is not positive semi-definite. Consider the $2 \times 2$ principal minor of $M$ given by the matrix
$$A =   \left[ \begin{array}{cc} M_{(i,i)} & M_{(i,j)} \\ M_{(j,i)} & M_{(j,j)} \end{array} \right].$$
We now compute
$$det(A) = M_{(i,i)}M_{(j,j)} - M_{(i,j)}M_{(j,i)} = M_{(i,i)}\cdot 0 - M_{(i,j)}\cdot M_{(i,j)}^* < 0.$$
Since $A$ has a principal minor with negative determinant, we conclude that $M$ has a negative eigenvalue, as desired.

\bigskip

In summary, we have shown that the row player has incentive to deviate unless
$$\sigma =   \left[ \begin{array}{c|c} \sigma_1 & 0 \\ \hline 0 & \frac{1}{2}\tilde{\rho} \end{array} \right].$$
We can derive analogous constraints on the density matrices by looking at when the column player has incentive to deviate.

\subsection{Impossibility of Satisfying Density Constraints}\label{impossibility}

Suppose that quantum state $\ket{\psi}$ is a canonical QCE implementation of the desired outcome distribution. We write
$$\ket{\psi} = \sum_{\substack{ x \in \{0,1\}^{n-1}\\ y \in \{0,1\}^{m-1}}} a_{xy} \ket{0x}\ket{1y} +  \sum_{\substack{ x \in \{0,1\}^{n-1}\\ y \in \{0,1\}^{m-1}}}  b_{xy} \ket{1x}\ket{0y} +  \sum_{\substack{ x \in \{0,1\}^{n-1}\\ y \in \{0,1\}^{m-1}}} c_{xy} \ket{1x}\ket{1y} .$$

Following our analysis from Appendix~\ref{generalization}, we will study the density matrices
\begin{align*}
\rho = 3 & \sum_{y} \sum_{x,x'} b_{xy}b_{x'y}^* \ket{1x}\bra{1x'}\\
\sigma = \frac{3}{2} & \left( \sum_{y} \sum_{x,x'} a_{xy}a_{x'y}^* \ket{0x}\bra{0x'} + a_{xy}c_{1x'1y}^* \ket{0x}\bra{1x'} + c_{xy}a_{x'y}^*\ket{1x}\bra{0x'} \right. \\
& + c_{xy}c_{x'y}^* \ket{1x}\bra{1x'}  \Bigg)  
\end{align*}
where the constants of $3$ and $\frac{3}{2}$ come from normalizing the density matrix after we take the conditional probability.

We now apply Lemma~\ref{densitycondition} to observe that
$$\frac{3}{2} \sum_{y}\sum_{x,x'} a_{xy}c_{x'y}^* \ket{0x}\bra{1x'}$$
must be a matrix of entirely zeroes (corresponding to the top-right block of $\sigma$). Therefore, we have the constraint
$$
\forall x, x': \sum_y a_{xy}c_{x'y}^* = 0.
$$
Furthermore, by comparing the bottom-right blocks of $\sigma$ and $\rho$, we have the constraint
$$\forall x, x': \sum_y b_{xy}b_{x'y}^* = \sum_y c_{xy}c_{x'y}^*.$$
We now apply the result analogous to Lemma~\ref{densitycondition} for the column player. This yields the constraints
$$\forall y, y': \sum_x b_{xy}c_{xy'}^* = 0$$
$$\forall y, y': \sum_x a_{xy}a_{xy'}^* = \sum_x c_{xy}c_{xy'}^*.$$
Finally, since $\ket{\psi}$ implements our desired outcome distribution, we have
$$\sum_{x,y}|a_{xy}|^2 = \sum_{x,y}|b_{xy}|^2 = \sum_{x,y}|c_{xy}|^2 = \frac{1}{3}.$$
The state $\ket{\psi}$ is a canonical QCE implementation of the desired outcome distribution if and only if all of the above conditions are satisfied.

We claim that the above set of equations has no solution. Define the matrices $A$, $B$, and $C$, where $A_{ij} = a_{ij}$, $B_{ij} = b_{ij}$, and $C_{ij} = c_{ij}$. We rewrite the above constraints as:

\begin{eqnarray*}
Tr(AA^\dagger) = Tr(BB^\dagger) = Tr(CC^\dagger) =& \frac{1}{3}\\
AC^\dagger =& 0\\
B^\dagger C =& 0 \label{BC} \\
BB^\dagger =& C C^\dagger\\
A^\dagger A =& C^\dagger C.
\end{eqnarray*}
We now compute
$$(CC^\dagger)(CC^\dagger) = C(C^\dagger C)C^\dagger = C(A^\dagger A)C^\dagger = (CA^\dagger)(AC^\dagger)=0.$$
We claim that the above equation implies that $CC^\dagger = 0$. Indeed, let $v$ be an eigenvector of the hermitian matrix $CC^\dagger$ with eigenvalue $\lambda$. We now compute
$$v^T (CC^\dagger)(CC^\dagger)v =  v^T(CC^\dagger)v \cdot \lambda = \lambda^2 v^T v = 0$$
and thus $\lambda = 0$. Since $v$ was an arbitrary eigenvector, we conclude that $CC^\dagger$ is the all-zero matrix. But this contradicts the fact that $Tr(CC^\dagger) = 1/3$.
 
 Therefore, there does not exist a state $\ket{\psi}$ satisfying the density matrix conditions of Appendix~\ref{generalization}, and thus there is no QCE achieving the outcome distribution $\frac{1}{3}(TR+BL+BR)$.

\section{The Complete Information GHZ Game}\label{ghz}

To prove Theorems~\ref{ghztheorem} and \ref{quantumefce}, we use a complete-information version of the GHZ game from \cite{ghz}. The standard, incomplete information GHZ game has three players: Alice, Bob, and Charlie. They are given input bits $a$, $b$, and $c$ respectively, with the promise that $a+b+c = 0 \pmod{2}$. The players output bits $x$, $y$, and $z$ respectively, and they win if $x+y+z \pmod{2} = a \vee b \vee c$.

It is straightforward to show that no classical strategy allows the three players to win with certainty. However, if they share an entangled state, then they can always win. In particular, suppose that they share the state
$$\ket{\psi_{GHZ}} = \frac{1}{2}(\ket{000} - \ket{011} - \ket{101} - \ket{110}),$$
where the first bit belongs to Alice, the second bit to Bob, and the third bit to Charlie. It is easy to check that the players always win if they all use the quantum protocol ``Apply a Hadamard transformation to your qubit if and only if your input bit is 1. Then measure your qubit in the standard basis and play the result.''

We now consider a ``complete information GHZ game'' (denoted cGHZ) which has four players: Alice, Bob, Charlie, and Nate. In the cGHZ game, Nate moves first, and has four possible moves (corresponding to each assignment of bits $a$, $b$, and $c$ such that $a+b+c=0 \mod{2}$.) Alice, Bob, and Charlie then move in turn, where each of these players can only distinguish between information sets based on their own input bit value. Each of these players has two possible moves (corresponding to the value they choose for their output bit) and we give each of Alice, Bob, and Charlie payoff 1 if their moves correspond to a winning outcome of the GHZ game, and payoff 0 otherwise. We give Nate a payoff equal to 1 minus Alice's payoff.

It is clear that there is no EFCE or IR-EFCE in which Alice, Bob, and Charlie have expected utility 1. (In particular, Nate can always guarantee himself non-zero expected utility by mixing randomly between his four actions.) In both of these equilibrium concepts, the advice at each information set is determined before the game begins, and for any fixed set of advice, Nate could get non-zero expected utility by mixing.

However, consider the QCE in which Alice, Bob, and Charlie share the state $\ket{\psi_{GHZ}}$ as before (where Nate holds none of the qubits from $\ket{\psi_{GHZ}}$), where Nate mixes uniformly between his 4 actions, and where the other three players act according to the winning strategy of the GHZ game. It is clear that this is a QCE: Nate cannot improve his utility by deviating, since Alice, Bob, and Charlie win regardless of Nate's move. The other three players have no incentive to deviate, since their expected utility of 1 is maximal. Since the outcome distribution of this QCE has expected utility 1 for Alice, Bob, and Charlie, there is no EFCE or IR-EFCE achieving this distribution. This completes the proof of Theorem~\ref{ghztheorem}.

\section{Proof of Theorem~\ref{quantumefce}}\label{finalproof}

We will now prove Theorem~\ref{quantumefce} by showing an extensive form game $G$ and outcome distribution $D$ which can be achieved by a QCE and by an EFCE but not by any IR-EFCE. The game $G$ is a five-player game combining  aspects of the cGHZ game with the game from Figure~\ref{EFCEnoQCE}. We construct this game by beginning with the structure from Figure~\ref{EFCEnoQCE}, and leave the $IN$ branch unchanged. (We give players 3, 4, and 5 payoffs of 0 in the outcomes $(IN,a,L), (IN,a,R), (IN,b,L)$, and $(IN,b,R)$.) In the $OUT$ branch, however, instead of having payoff $(3,3)$, we have a version of cGHZ game, where player 1 takes the role of  Nate and players 3, 4, and 5 have the roles of Alice, Bob, and Charlie. (Player 2 gets no moves in this branch of the tree, and receives payoff 0 in all of the outcomes.) In this version of the cGHZ game, players 3, 4, and 5 get payoff 1 if they succeed and player 1 gets payoff 0. If players 3, 4, and 5 fail in the cGHZ game, then they get 0 payoff while player 1 gets payoff 50.

Consider the outcome distribution $D = 1/2(IN,a,L) + 1/2(IN,b,R)$. I claim that $D$ can be achieved by a QCE and by an EFCE but not by any IR-EFCE. First, we'll show that there is a QCE with outcome distribution $D$. Consider the QCE where players 3, 4, and 5 share the entangled state $\ket{\psi_{GHZ}}$ (and they are instructed to use the appropriate circuits from the GHZ protocol) and players 1 and 2 share the state $\frac{1}{\sqrt{2}}(\ket{0}\ket{0} + \ket{1}\ket{1})$. We instruct player 1 to play $IN$, and then in his next information set to measure his qubit in the standard basis. Player 2 is instructed (if he has the opportunity to move) to measure his qubit in the standard basis and to play the corresponding move.

If everyone follows their prescribed protocol, then it is clear that the resulting outcome distribution is indeed $D$. I now claim that no player can benefit by deviating. In particular, players 3, 4, and 5 never have the opportunity to move (since player 1 is playing $IN$) and thus have no incentive to deviate. Player 1 realizes that, if he were to play $OUT$ instead of $IN$, that he would receive payoff 0 (since players 3, 4, and 5 will always win the cGHZ game), and thus he will indeed play $IN$. Once he plays $IN$, it is in his interest for his next action to coordinate with player 2. Similarly, it is obvious that player 2 has no incentive to deviate.

Furthermore, it is clear that $D$ can be achieved by an EFCE. The key point is that player 1's expected utility of following his advice and playing $IN$ in his first information set is 51, while he can never get utility more than 50 by playing $OUT$. The remainder of the argument is analogous to the argument above.

We finally claim that there is no IR-EFCE achieving distribution $D$. The reason is that player 1 can obtain expected utility at least $\frac{50}{4}>2$ by playing $IN$ and then mixing randomly between his four available actions in the cGHZ game. Therefore, if he sees that his later advice will be $R$, he can improve his utility immediately by deviating and playing $OUT$.

\end{document}